\documentclass[conference]{IEEEtran}
\IEEEoverridecommandlockouts
\usepackage{cite}
\usepackage{amsmath,amssymb,amsfonts}
\usepackage{algorithmic}
\usepackage{graphicx}
\usepackage{textcomp}
\def\BibTeX{{\rm B\kern-.05em{\sc i\kern-.025em b}\kern-.08em
    T\kern-.1667em\lower.7ex\hbox{E}\kern-.125emX}}

\usepackage{amsmath}
\usepackage{amsthm}
\usepackage{amssymb}
\usepackage{bm}
\usepackage{xspace}
\usepackage{xcolor}
\usepackage{graphicx}
\usepackage{url}
\usepackage{framed}
\usepackage{float}
\usepackage{rotating}
\usepackage{verbatim}
\usepackage{listings}
\usepackage{lscape}
\usepackage[ruled,vlined,linesnumbered]{algorithm2e}
\usepackage{caption}
\usepackage[font=footnotesize]{caption}
\usepackage{csquotes}

\usepackage{pgfplots}
\usepackage{pgf}
\usepackage{tikz}
\usetikzlibrary{arrows,shapes.misc,chains,scopes}
\usetikzlibrary{decorations.pathreplacing}

\usepackage{pgfplotstable}

\usepackage{colortbl}
\usetikzlibrary{matrix}
\usetikzlibrary{calc}
\usetikzlibrary{fit}

\newcommand{\executeiffilenewer}[3]{%
	\ifnum\pdfstrcmp{\pdffilemoddate{#1}}%
	{\pdffilemoddate{#2}}>0%
	{\immediate\write18{#3}}\fi%
}

\graphicspath{{figures/}}

\newcommand{\bmm}{\begin{matrix}}
	\newcommand{\emm}{\end{matrix}}
\newcommand{\bpm}{\begin{pmatrix}}
	\newcommand{\epm}{\end{pmatrix}}

\newcommand{\bsbm}{\left[\begin{smallmatrix}}
	\newcommand{\esbm}{\end{smallmatrix}\right]}

\newcommand{\bbm}{\begin{bmatrix}}
	\newcommand{\ebm}{\end{bmatrix}}

\usepackage[utf8]{inputenc}
\usepackage[english]{babel}

\usepackage{booktabs}
\usepackage{multirow}
\usepackage{siunitx}
\usepackage{smartdiagram}

\theoremstyle{definition}
\newtheorem{theorem}{Theorem}

\newcommand{\midarrow}{\tikz \draw[-triangle 60] (0,0) -- +(-0.1,0);}
\newcommand{\midrarrow}{\tikz \draw[-triangle 60] (0,0) -- +(0.1,0);}
\newcommand{\miduarrow}{\tikz \draw[-triangle 60] (0,0) -- +(0,-0.1);}

\begin{document}

\title{Flexible IR-HARQ Scheme for Polar-Coded Modulation}

\author{
	\IEEEauthorblockN{Peihong Yuan, Fabian Steiner, Tobias Prinz}
	\IEEEauthorblockA{\textit{Institute for Communications Engineering} \\
	\textit{Technical University of Munich}\\
	Munich, Germany \\
	\{peihong.yuan,fabian.steiner,tobias.prinz\}@tum.de}
\and
	\IEEEauthorblockN{Georg B\"ocherer}
	\IEEEauthorblockA{\textit{Mathematical and Algorithmic Sciences Lab} \\
	\textit{Huawei Technologies France}\\
	Paris, France \\
	georg.boecherer@ieee.org}
}

\maketitle

\begin{abstract}
	A flexible incremental redundancy hybrid automated repeat request (IR-HARQ) scheme for polar codes is proposed based on dynamically frozen bits and the quasi-uniform puncturing (QUP) algorithm. The length of each transmission is not restricted to a power of two. It is applicable for the binary input additive white Gaussian noise (biAWGN) channel as well as higher-order modulation. Simulation results show that this scheme has similar performance as directly designed polar codes with QUP and outperforms LTE-turbo and 5G-LDPC codes with IR-HARQ. 
\end{abstract}


\section{Introduction}
Many communication channels are time-varying and unknown to the transmitter. Incremental redundancy hybrid automated repeat request (IR-HARQ) as shown in Fig.~\ref{fig:harq} is a scheme that transmits additional redundancy bits until the data bits can be reconstructed. For turbo codes (such as those used in LTE), a low rate mother code is punctured with different patterns for several transmissions. The coding scheme for enhanced mobile broadband (eMBB) in 5G uses protograph-based, Raptor-like LDPC codes \cite{chen_protograph-based_2015} that allow for both flexible block length and code rate adaptation. The standard defines two base matrices that offer optimized performance for different operating regimes.

With cyclic redundancy check (CRC) outer codes and successive cancellation list (SCL) decoding \cite{tal2015list}, polar codes \cite{stolte2002rekursive,arikan2009channel} outperform state-of-the-art turbo and LDPC codes in the short to medium length regime. Polar-coded modulation (PCM) is discussed in \cite{seidl2013polar,mahdavifar2016polar,tavildar2017bit}. The performance comparison and efficient code design methods of three polar-coded modulation schemes are presented in \cite[Fig.~11]{bocherer2017efficient}. Multilevel polar coding (MLPC) with set partitioning (SP) labeling in \cite{seidl2013polar} performs best and is around 1 dB more power efficient than an AR4JA \cite{divsalar2009capacity} LDPC code decoded with 200 iterations.

A quasi-uniform puncturing (QUP) algorithm was proposed in \cite{niu2013beyond} to efficiently design length-flexible polar codes, i.e., polar codes where the number of coded bits is not limited to be a power of two. 

In \cite{li2016capacity}, a scheme the authors called \enquote{polar codes with incremental freezing} and in \cite{hong2016capacity} \enquote{parallel-concatenated polar codes} is proposed. The capacity-achievability of this scheme is proved in \cite{li2016capacity} by using the nesting property. In \cite{ma2017incremental}, a polar code extension method is presented which outperforms the scheme in \cite{li2016capacity,hong2016capacity} for finite block length. 

In this work, we prove that the scheme in \cite{ma2017incremental} can achieve capacity asymptotically in the block length under some design constraints. In addition, a length-flexible IR-HARQ scheme based on dynamically frozen bits and QUP is proposed. This scheme is extended to polar-coded modulation with amplitude shift keying (ASK) and quadrature amplitude modulation (QAM) constellations. Simulation results show that the polar codes designed by our algorithm have similar error correction performance as directly designed polar codes.

This work is organized as follows. In Sec.~\ref{sec:pre}, we review polar codes, PCM and QUP. We discuss existing and proposed IR-HARQ schemes in Sec.~\ref{sec:harq}. Sec.~\ref{sec:sim} provides design examples and numerical results. The performance of the proposed scheme is compared with directly designed polar codes and 4G/5G codes \cite{3gpp.36.212,qualcomm_codes} in additive white Gaussian noise (AWGN) channels. We conclude in Sec.~\ref{sec:con}.

\begin{figure}
	\footnotesize
	\centering
	\[
	\renewcommand{\arraystretch}{1.5}
	\begin{array}{|c|c|c|c|}
	\multicolumn{1}{c}{n_1} & \multicolumn{1}{c}{n_2} & \omit & \multicolumn{1}{c}{n_t} \\
	\omit\mathstrut\downbracefill & \omit\mathstrut\downbracefill & \omit & \omit\mathstrut\downbracefill \\[.2em]
	\hline
	\small\text{$1$st Transmission} & \small\text{$2$nd Transmission}  & \small\cdots & \small\text{$t$th Transmission}\\
	\hline
	\omit\mathstrut\upbracefill & \omit & \omit & \omit \\
	\multicolumn{1}{c}{\text{$1$st decoding}} & \omit & \omit & \omit \\
	\omit \span\omit\mathstrut\upbracefill & \omit & \omit \\
	\multicolumn{2}{c}{\text{$2$nd decoding}} & \omit & \omit \\
	\omit \span\omit\span\omit\span\omit\mathstrut\upbracefill \\
	\multicolumn{4}{c}{\text{$t$th decoding}}  \\
	\end{array}
	\]	
	\caption{IR-HARQ.}
	\label{fig:harq}
\end{figure}

\section{Preliminaries}\label{sec:pre}

\subsection{Polar Coding}
In this paper, uppercase letters $X,Y,B,U$ denote the random variables (RV) while the corresponding lowercase letters are their realizations. The notation $c_1^N$ is short for $c_1c_2\dots c_N$. For an arbitrary subset $\mathcal{S}$ of $\left\{1,\dots,N\right\}$, $\mathcal{S}^\complement$ is the complement of $\mathcal{S}$ and $c_\mathcal{S}$ denotes the vector of  $c_1^N$ formed by the elements with indices in $\mathcal{S}$.

A binary polar code of block length $N$ and dimension $k$ is defined by the polar transform with matrix $\mathbb{F}^{\otimes\log_2N}$ and $N-k$ frozen positions, where $\mathbb{F}^{\otimes\log_2N}$ denotes the $\log_2 N$-fold Kronecker power of the kernel 
\begin{equation}
\mathbb{F}= \begin{bmatrix}
1 & 0\\
1 & 1
\end{bmatrix}.
\end{equation}
Polar encoding can be represented by
\begin{equation}\label{eq:polartrans}
c_1^N=u_1^N \mathbb{F}^{\otimes\log_2N}.
\end{equation}
The vector $c_1^N$ denotes the codeword. The vector $u_1^N$ includes $k$ information bits $u_\mathcal{A}$ and $N-k$ predefined frozen bits $u_{\mathcal{A}^\complement}$. $\mathcal{A}$ and $\mathcal{A}^\complement$ are called the information set and frozen set defined in \cite{arikan2009channel}. SC decoding uses the channel observation and previous estimates $\hat{u}_1,\dots, \hat{u}_{i-1}$ to decode $u_i$. Both encoding and SC decoding have complexity $\mathcal{O}\left(N\log_2 N\right)$ \cite{arikan2009channel}.

The polar code construction finds the most reliable bits in $u_1^N$ under SC decoding. The Monte Carlo (MC) construction was introduced in \cite{stolte2002rekursive,arikan2009channel}, and needs extensive simulations. An information theoretical construction was introduced in \cite{arikan2009channel}. The reliability of the $i$th bit can be quantified by the mutual information (MI) ${\rm I}\left(U_i;Y^N|U_1^{i-1}\right)$. We can calculate these MIs for all $i \in \left\{1,2,\dots, N\right\}$ by recursively calculating the MIs of the basic polar transform displayed in Fig.~\ref{fig:basic_polartrans}. 

For the binary input additive white Gaussian noise (biAWGN) channel, density evolution \cite{mori2009performance} with Gaussian approximation (GA) \cite{ten2004design} has much lower complexity and performs very close to the MC construction. The update rule for the basic polar transform is given by
\begin{align}
I^-&=1-J \left( \sqrt[]{ \left[J^{-1}(1-I_1)\right]^2 + \left[J^{-1}(1-I_2)\right]^2 } \right)\\
I^+&=J \left( \sqrt[]{ \left[J^{-1}(I_1)\right]^2 + \left[J^{-1}(I_2)\right]^2 } \right)
\end{align}
The numerical approximations in \cite{brannstrom2005convergence} can be used for $J(\cdot)$ and $J^{-1}(\cdot)$.
Additionally, the frame error rate (FER) under SC decoding can be estimated by 
\begin{equation}
{\rm Pr}\left( \hat{U}_i \neq U_i | \hat{U}_1^{i-1} = U_1^{i-1} \right) = Q\left(\frac{1}{2}J^{-1}\left( {\rm I}\left(U_i;Y^n|U_1^{i-1}\right) \right)\right) 
\end{equation}
\begin{equation}
\text{FER}_\text{SC,est} = 1-\prod_{i \in \mathcal{A}}\left( 1- {\rm Pr}\left( \hat{U}_i \neq U_i | \hat{U}_1^{i-1} = U_1^{i-1} \right) \right)
\end{equation}
where
\begin{equation}
Q(x) = \frac{1}{\sqrt[]{2\pi}}\int_{x}^{\infty}e^{\frac{u^2}{2}}{\rm d}u.
\end{equation}

\begin{figure}
	\centering
	\footnotesize
	\begin{tikzpicture}
	\draw (0,0)--(3,0); 
	\draw (0,1)--(3,1);
	\draw (1.5,0)--(1.5,1);
	\draw[fill=white] (1.5,1) circle [radius=0.2cm];
	\node[color=black]at (1.5,1) {$+$};
	\node at (0,0.25) {$I^+={\rm I}(U_2;Y_1Y_2|U_1)$};
	\node at (0,1.25) {$I^-={\rm I}(U_2;Y_1Y_2)$};
	\node at (3,0.25) {$I_2={\rm I}(B_2;Y_2)$};
	\node at (3,1.25) {$I_1={\rm I}(B_1;Y_1)$};
	\end{tikzpicture}
	\caption{MIs of the basic polar transform}
	\label{fig:basic_polartrans}
\end{figure}
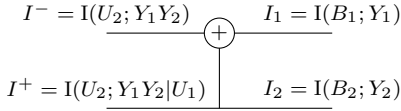

 The nesting property is introduced in \cite{li2016capacity}. The authors show that the reliability order of the polarized bits is independent on the channel quality for infinite length. Note that the nesting property generally does not hold for finite length.

\subsection{Polar-Coded Modulation}
We consider the following memoryless AWGN channel model with $2^m$-ASK constellations $X\in \mathcal{X} = \left\{ \pm 1,\pm 3,\dots,\pm (2^m-1) \right\}$.
\begin{equation}
Y=X+\sigma Z
\end{equation}
where $Z$ is zero mean Gaussian noise with variance one. Note that ASK constellations can be orthogonally extended to QAM constellations. 

Polar-coded modulation (PCM) schemes are discussed in \cite{seidl2013polar,mahdavifar2016polar,tavildar2017bit}. In \cite{bocherer2017efficient}, the performance of three schemes is compared and their efficient design algorithms are presented. MLPC with SP labeling in \cite{seidl2013polar} provides the best performance. MLPC with code length $mN$ works as following: 
\begin{itemize}
	\item Encoding
				\begin{itemize}
					\item[1.] Put $k$ information bits in vector $u_1^{mN}$ and define
						\begin{equation}u_{1,j}^N= {u_{\left(j-1\right)N+1}^{jN}},~j = 1,\dots,m.
						\end{equation}
					\item[2.] Encode $m$ polar codes
						\begin{equation}c_{1,j}^N= u_{1,j}^N \mathbb{F}^{\otimes\log_2N},~j = 1,\dots,m.
						\end{equation}
					\item[3.] Map the code words to symbols for $i = 1,\dots,N$
						\begin{equation}
						\begin{split}
						c_{i,1},c_{i,2},\dots,c_{i,m}&= b_{i,1},b_{i,2},\dots,b_{i,m} \\
						&=(b_1,b_2,\dots,b_m)_i \overset{\text{SP}}{\mapsto} x_i.
						\end{split}
						\end{equation}
			   \end{itemize}
	\item Decoding
				\begin{itemize}
					\item[1.] Demap and decode level 1
					\begin{equation}
					\begin{split}
						\ell_{i,1} &= \log\frac{p_{B_{i,1}|Y_i}(0|y_i)}{p_{B_{i,1}|Y_i}(1|y_i)},~ i = 1,\dots,N \\
						\hat{c}_{1,1}^N &= \texttt{polarSCDecode}\left(l_{1,1},\dots,l_{N,1}\right).
					\end{split}
					\end{equation}
					\item[$j$.] Demap and decode level $j$ for $j = 2,\dots,m$
					\begin{equation}
					\begin{split}
					\ell_{i,j} &= \log\frac{p_{B_{i,j}|Y_i B_{i,1} \dots B_{i,j-1} }\left(0|y_i \hat{c}_{i,1} \dots \hat{c}_{i,j-1}\right)}{p_{B_{i,j}|Y_iB_{i,1} \dots B_{i,j-1}}\left(1|y_i\hat{c}_{i,1} \dots \hat{c}_{i,j-1}\right)},\\&\qquad \qquad \qquad \qquad \qquad \qquad i = 1,\dots,N \\
					\hat{c}_{1,j}^N &= \texttt{polarSCDecode}\left(l_{1,j},\dots,l_{N,j}\right)
					\end{split}
					\end{equation}
				\end{itemize}
\end{itemize}
where $l$ denotes the input log-likelihood ratio (LLR).

MI demapper GA (MI-DGA) construction calculates ${\rm I}\left(B_j;YB_1^{j-1}\right)$ for $j=1,\dots,m$ and uses them to find the most reliable bits in $u_1^{mN}$ with GA. Therefore, the required code rates of the $m$ polar codes $k_1/N,\dots,k_m/N$ are given by MI-DGA. Note that the overall code rate of the system is $\left(k_1+\cdots+k_m\right)/mN = k/mN$ and the transmission rate is $k/N ~ \text{bits/channel use}$.

\subsection{Rate-Matched Polar Codes}
Because of the recursive structure of $\mathbb{F}^{\otimes\log_2N}$, polar codes usually have a block length that is a power of two. Punctured polar codes are introduced in \cite{eslami2011practical, wang2014novel}. The punctured $\left(n,k\right)$ polar codes can be decoded with standard polar decoders for length $N=2^{ \left \lceil{\log_2 n }\right \rceil }$. For punctured polar codes, $N-n$ bits in $c_1^N$ are not transmitted and the corresponding LLRs are set to zero. For GA construction, the initial MIs of the punctured bits are zero. 

With the QUP algorithm, the first $N-n$ bits in $c_1^N$ are not transmitted, i.e.,
\begin{align*}
\text{GA construction:}&~\text{set} ~{\rm I}\left(B_i;Y\right)=0,~i=1,\dots,N-n\\
\text{after encoding:}&~\text{transmit} ~c_{N-n+1}^N\\
\text{before decoding:}&~\text{set} ~l_1^{N-n} = 0.
\end{align*}
This algorithm is called QUP because the punctured position in bit-reversal representation \cite{arikan2009channel} looks like uniformly distributed in $\left\{1,\dots,N\right\}$. Let $\left(n,k,N\right)$ denote a QUP polar code with dimension $k$, block length $n$ and punctured from $N$ bits mother codes with QUP algorithm, where $N$ has to the power of two.

\begin{theorem} \label{th1}
	For an $\left(n,k,N\right)$ QUP-polar code, the first $N-n$ bits in $u_1^N$ are frozen.
\end{theorem}
\begin{proof}
For the MI update of the basic polar transform in Fig.~\ref{fig:basic_polartrans}, it is easy to show that
\begin{align}
I^- \leq \min \{ I_1, I_2 \},~I^+ \geq \max \{ I_1, I_2 \}.
\end{align}
As MI is always non-negative, we have
\begin{align}
I^- =0,~&\text{if}~I_1 = 0~\text{or}~I_2 = 0 \\
I^- =I^+=0,~&\text{if}~I_1 =I_2= 0.
\end{align}
Thus the number of channels with zero capacity is invariant. With QUP we set $I_1^{N-n}$ to zero. Because of the recursive structure of $\mathbb{F}^{\otimes\log_2N}$ the zeros will propagate through the transform which causes ${\rm I}\left(U_i;Y^N|U_1^{i-1}\right)=0$ for $i=1,\dots,N-n$.
\end{proof}

\begin{theorem} \label{th2}
	All $\left(n,k,2^jN\right)$ QUP-polar codes have the same encoding and decoding complexity as $\left(n,k,N\right)$ QUP-polar codes, where $j$ is a natural number and $N=2^{ \left \lceil{\log_2 n }\right \rceil }$.
\end{theorem}
\begin{proof}
	The decoder for $\left(n,k,2N\right)$ QUP-polar codes is shown in Fig.~\ref{fig:punc_polar}. Obviously, more than $N$ bits are punctured. Thus, $u_1^N$ are all frozen because of Theorem~\ref{th1}. According to the SC decoding, we first decode $u_1^N$ (upper decoder) and then decode $u_{N+1}^{2N}$ (lower decoder) based on $\hat{u}_1^N$. The input of the lower decoder is $(1-0)\cdot\ell_1^N+\ell_{N+1}^{2N}=\ell_{N+1}^{2N}$. Thus for this QUP-polar code, we just need to run the lower decoder (dashed box in Fig.~\ref{fig:punc_polar}) which is the decoder for $\left(n,k,N\right)$ QUP-polar codes. With the same idea we can extend the theorem to mother code length $2^jN$.
	

\end{proof}

\begin{figure}
	\centering
	\footnotesize
	\begin{tikzpicture}
	\draw  (0,0) rectangle (2,2);
	\node at (1,1) {$F^{\otimes \log_2N}$};
	\draw  (0,2.5) rectangle (2,4.5);
	\node at (1,3.5) {$F^{\otimes \log_2N}$};
	\draw (2,1)--node {\midarrow}(5,1); 
	\draw (2,3.5)--node {\midrarrow}(5,3.5);
	\draw (5,1)--node {\miduarrow}(5,3.5);
	\draw (6,1)--node {\midarrow}(5,1);
	\draw (6,3.5)--node {\midarrow}(5,3.5);
	\draw[fill=white] (5,3.5) circle [radius=0.2cm];
	\node[color=black]at (5,3.5) {$+$};
	\node at (6.25,1.5) {$\ell_{N+1}^{2N}$};
	\node at (6.25,4) {$\ell_{1}^{N}=0$};
	\node at (-0.5,1) {$u_{N+1}^{2N}$};
	\node at (-0.6,3.5) {$u_{1}^{N}=0$};
	\node at (3.5,4) {$ 0$};
	\node at (4.5,2.25) {$0$};
	\node at (3.5,1.5) {$\ell_{N+1}^{2N}$};
	\draw [rounded corners, dashed] (-1.2,-0.1) rectangle (4,2.1);
		\draw [rounded corners, dotted] (-1.2,2.4) rectangle (4,4.6);
	\node at (3.5,0.1) {lower};
	\node at (3.5,2.6) {upper};
	\end{tikzpicture}
	\caption{Equivalence of $\left(n,k,N\right)$ and $\left(n,k,2N\right)$ QUP-polar code.}
	\label{fig:punc_polar}
\end{figure}
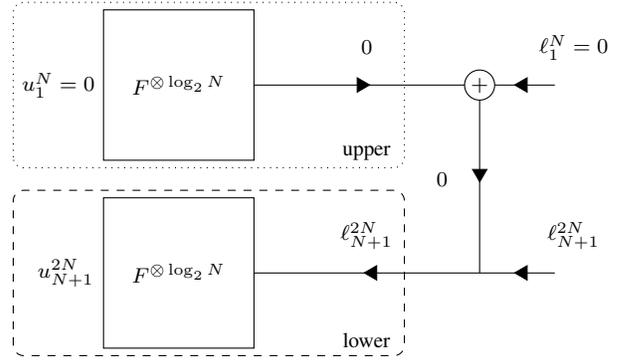

\section{IR-HARQ with Polar Codes}\label{sec:harq}
The basic idea of IR-HARQ is displayed in Fig.~\ref{fig:harq}. The decoder receives $n_t$ bits from the $t$th transmission and then decodes $\left(\sum_{q=1}^{t}n_q,k\right)$ code. Using the CRC, the receiver may detect a decoding failure, in which case it requests the $\left(t+1\right)$th transmission from the sender.

\subsection{Existing Schemes: Polar Codes with Incremental Freezing}
In \cite{li2016capacity,hong2016capacity}, polar codes with incremental freezing are proposed. The main idea is as follows. An $\left(N_1,k\right)$ polar code is transmitted first. When retransmission occurs, we additionally freeze $k^\prime$ \enquote{most unreliable} information bits in the $\left(N_1,k\right)$ code and send them with an $\left(N_2,k^\prime\right)$ code in next transmission. After the $2$nd transmission we decode the $\left(N_2,k^\prime\right)$ code and the first code successively. Note that the first code becomes a $\left(N_1,k-k^\prime\right)$ code with the estimation of the \enquote{most unreliable} information bits. For the $3$rd transmission (if needed), we freeze $k^{\prime\prime}$ \enquote{most unreliable} information bits in both the $\left(N_1,k-k^\prime\right)$ code and the $(N_2,k^\prime)$ code. The retransmissions are continued with the same manner until the decoding is successful. For infinite code length, this scheme achieves capacity \cite{li2016capacity} with the nesting property. We note that this scheme is equivalent to dividing the $\left(\sum_{q=1}^{t}N_q,k\right)$ code into $t$ separated polar codes and decoding them successively, which causes a huge performance loss in the finite length regime.

\subsection{Existing Schemes: Polar Extension}
A polar code extension method is presented in \cite{ma2017incremental}. Consider the $1$st transmission with $\left(N_1,k\right)$ polar code $c(1)_1^{N_1}=u(1)_1^{N_1} \mathbb{F}^{\otimes\log_2{N_1}}$, where $c(t)_1^N$ denotes the codeword after the $t$th transmission. When retransmission occurs, we design a $(2N_1,k)$ polar code $c(2)_1^{2{N_1}}=u(2)_1^{2{N_1}} \mathbb{F}^{\otimes\log_2(2{N_1})}$, let the information set in $u(2)_{{N_1}+1}^{2{N_1}}$ be a subset of the information set in $u(1)_1^{N_1}$ and set $u(2)_{{N_1}+1}^{2{N_1}} = u(1)_1^{N_1}$. Then copy the bits, which are reliable in $u(1)_1^{N_1}$ but unreliable in $u(2)_{{N_1}+1}^{2{N_1}}$ to the most reliable positions in $u(2)_1^{N_1}$. The code word can be presented by 
 \begin{align}
c(2)_{{N_1}+1}^{2{N_1}} &= c(1)_1^{N_1} \\
c(2)_{1}^{{N_1}} &= c(1)_1^{N_1} \oplus u(2)_1^{N_1}\mathbb{F}^{\otimes\log_2{N_1}}.
 \end{align}
Thus, only the first ${N_1}$ bits $c(2)_{1}^{{N_1}}$ need to be transmitted, because the receiver already knows the remaining ${N_1}$ bits of the $\left(2{N_1},k\right)$ polar code from the previous transmission. The receiver first decodes the information bits in $u(2)_1^{N_1}$, copies them to $u(2)_{{N_1}+1}^{2{N_1}}$ as frozen bits and then decodes the remaining bits. The only difference between a directly generated $\left(2{N_1},k\right)$ polar code and the code after the $2$nd transmission is the information set in $u(2)_{{N_1}+1}^{2{N_1}}$ which has to be a subset of the information set in $u(1)_1^{N_1}$. In this scheme, we decode an $\left(N_2,k\right)$ code after the $2$nd transmission ($N_2=2N_1$). This extension is repeated until the decoding is successful. This extension method is not flexible because the transmission length has to be the same as the sum of all previous transmissions, i.e.,
\begin{equation}
N_t = \sum_{q=1}^{t-1}N_q,~t=2,3,\ldots
\end{equation}

This scheme is capacity-achieving under some constraints. According to the information theoretical construction in \cite{arikan2009channel}, designing a $\left(2N,k\right)$ polar code for a channel with mutual information $I$ is equivalent to designing a $(N,k_1)$ polar code  for $I^-$ and a $\left(N,k_2\right)$ polar code for $I^+$, where
  \begin{align}
  k_1+k_2 &=k \\
  k_1/k_2 &= I^-/I^+.
  \end{align}
 Let $I_t$ be the design MI for the $t$th transmission and $2N_t$ the code length after the $t$th transmission ($t\geq2$). Note that only $N_t$ bits are sent in the $t$th transmission. Under the constraints
 \begin{equation}
 I_{t}^+ \leq I_{t-1},~t=2,3,\dots
 \end{equation}
the information set in $u(t+1)_{N_t+1}^{2N_t}$ is always a subset of the information set in $u(t)_1^{N_t}$ because of the nesting property, i.e., every polar code after $t$th transmission is an optimal $\left(\sum_{q=1}^{t}N_q,k\right)$ polar code. 

\subsection{Proposed Scheme}
Polar codes with dynamically frozen bits are proposed in \cite{trifonov2016polar} to improve the distance properties of polar codes. The idea is to predetermine some frozen bits as linear combinations of previous information bits. The corresponding \enquote{dynamically frozen constraint} is needed to encode and decode. Our scheme is based on this technique and QUP. There are no constraints regarding the block length of any transmission.

\begin{figure}
	\centering
	\footnotesize
	\begin{tikzpicture}
	\fill[fill=gray!20] (0,0)--(1.5,0)--(0,1.5)--cycle;
	\fill[fill=gray!20] (1.5,0)--(3,0)--(1.5,1.5)--cycle;
	\fill[fill=gray!20] (3,0)--(4.5,0)--(3,1.5)--cycle;
	\fill[fill=gray!20] (4.5,0)--(6,0)--(4.5,1.5)--cycle;
	\fill[fill=gray!20] (0,1.5)--(1.5,1.5)--(0,3)--cycle;
	\fill[fill=gray!20] (3,1.5)--(4.5,1.5)--(3,3)--cycle;
	\fill[fill=gray!20] (0,3)--(1.5,3)--(0,4.5)--cycle;
	\fill[fill=gray!20] (1.5,3)--(3,3)--(1.5,4.5)--cycle;
	\fill[fill=gray!20] (0,0)--(1.5,0)--(0,1.5)--cycle;
	\fill[fill=gray!20] (0,4.5)--(1.5,4.5)--(0,6)--cycle;
	
	\draw [rounded corners,thick] (0,0) rectangle (6,6);
	\node at (5,5.5) {$F^{\otimes \log_2N}$};
	\draw [dashed, rounded corners] (-0.2,-0.1) -- (6.2, -0.1) -- (6.2,1.35)--   (-0.2,1.35)--cycle;
	\draw [decorate,decoration={brace,amplitude=5pt,mirror,raise=4pt},yshift=0pt] (6.2,0) -- (6.2,1.35) node [black,midway,xshift=0.7cm] {$n_1$};
	\draw [dashed, rounded corners] (-0.2,1.35) -- (6.2, 1.35) -- (6.2,2.5)--   (-0.2,2.5)--cycle;
	\draw [decorate,decoration={brace,amplitude=5pt,mirror,raise=4pt},yshift=0pt] (6.2,1.35) -- (6.2,2.5) node [black,midway,xshift=0.7cm] {$n_2$};
	\draw [dashed, rounded corners] (-0.2,2.5) -- (6.2, 2.5) -- (6.2,3.5)--   (-0.2,3.5)--cycle;
	\path (6.2,2.5) -- (6.2,3.5) node [black,midway,xshift=5pt,yshift=2pt] {$\vdots$};
	\draw [dashed, rounded corners] (-0.2,3.5) -- (6.2, 3.5) -- (6.2,5)--   (-0.2,5)--cycle;
	\draw [decorate,decoration={brace,amplitude=5pt,mirror,raise=4pt},yshift=0pt] (6.2,3.5) -- (6.2,5) node [black,midway,xshift=0.7cm] {$n_{t_{\text{max}}}$};
	
	\draw [decorate,decoration={brace,amplitude=5pt,raise=4pt},yshift=0pt] (-0.2,0) -- (-0.2,1.35) node [black,midway,xshift=-0.6cm] {$\mathcal{I}_1$};
	\draw [decorate,decoration={brace,amplitude=5pt,raise=4pt},yshift=0pt] (-0.2,1.35) -- (-0.2,2.5) node [black,midway,xshift=-0.6cm] {$\mathcal{I}_2$};
	\draw [decorate,decoration={brace,amplitude=5pt,raise=4pt},yshift=0pt] (-0.2,3.5) -- (-0.2,5) node [black,midway,xshift=-0.6cm] {$\mathcal{I}_{t_{\text{max}}}$};
	\path (-0.2,2.5) -- (-0.2,3.5) node [black,midway,xshift=-5pt,yshift=2pt] {$\vdots$};
	\end{tikzpicture}
	\caption{Proposed IR-HARQ scheme for polar coding}
	\label{fig:polar-harq}
\end{figure}
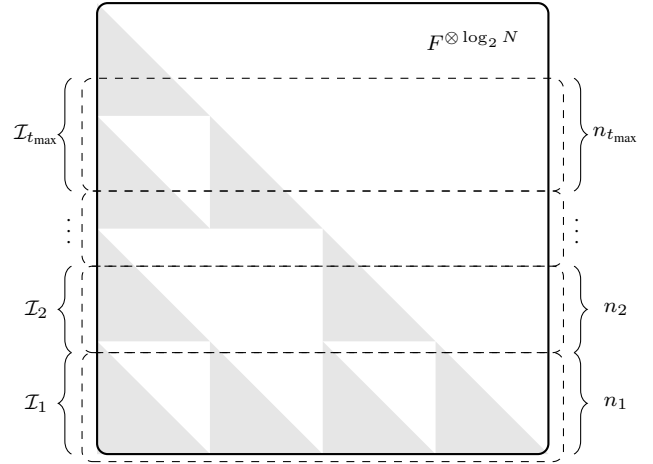

We assume the system is designed for a maximum of $t_{\text{max}}$ transmissions. Let $N$ be the length of a mother polar code
\begin{equation}
N = 2^{ \left \lceil{\log_2\left(\sum_{q=1}^{t_{\text{max}}}n_q\right)}\right \rceil },
\end{equation}
where $\left \lceil{\cdot}\right \rceil $ denotes the ceiling function and $n_q$ denotes the length of the $q$th transmission. In the proposed scheme, after the $t$th transmission, a $\left(\sum_{q=1}^{t}n_q,k,N\right)$ QUP-polar code is decoded. The main structure and design algorithm are displayed in Fig.~\ref{fig:polar-harq} and Algorithm~\ref{alg1} respectively. The set $\mathcal{I}_t$ is 
\begin{equation}
\left\{ N-\sum_{q=1}^{t}n_q+1, \dots, N-\sum_{q=1}^{t-1}n_q \right\},
\end{equation}
and $\mathcal{A}_t$ and $\mathcal{A}^\complement_t$ denote the information and frozen set of the $\left(\sum_{q=1}^{t}n_q,k,N\right)$ QUP-polar code after the $t$th transmission. The output dynamically frozen constraint is used for encoding and decoding. Note that $ \mathcal{A}_t \cup \mathcal{A}^\complement_t = \bigcup_{q=1}^t \mathcal{I}_q$. First ${N-\sum_{q=1}^{t}n_q}$ bits are frozen to zero (Theorem~\ref{th1}), but their indices are neither in $\mathcal{A}_t$ nor $\mathcal{A}^\complement_t$. 
\begin{algorithm}[]
	\footnotesize
	\DontPrintSemicolon
	\SetAlgoLined
	\SetKwInOut{Input}{Input}\SetKwInOut{Output}{Output}
	\Input{message length $k$, mother code length $N$, \\$t$th transmission length $n_t$, design MI $I_t$, \\previous code length $n_t^\prime=\sum_{q=1}^{t-1}n_q$, \\previous frozen set $\mathcal{A}^{\prime \complement}_{t} =\bigcup_{q=1}^{t-1} \mathcal{A}^\complement_q$, \\previous information set $\mathcal{A}^{\prime}_{t} =\bigcup_{q=1}^{t-1} \mathcal{A}_q$}
	\Output{information set $\mathcal{A}_t$, \\frozen set $\mathcal{A}^\complement_t$, \\dynamically frozen constraint}
	\BlankLine
	
Initialize design MI: $I_1^{N-n_t-n_t^\prime}=0, I_{N-n_t-n_t^\prime+1}^N=I_t$.\\
Estimate ${\text{Pe}_i} := {\rm Pr}\left( \hat{U}_i \neq U_i | \hat{U}_1^{i-1} = U_1^{i-1} \right)$.\\
Set $\text{Pe}_{\mathcal{A}^{\prime \complement}_{t}}=1$.
Find $k$ smallest in $\text{Pe}_{N-n_t-n_t^\prime+1}^N$ and put their indices in $\mathcal{A}_t$.\\
Frozen set $\mathcal{A}^\complement_t =  \bigcup_{q=1}^{t} \mathcal{I}_q \setminus \mathcal{A}_{t}$.\\
\If{$t\neq 1$}{
	Dynamically frozen constraint is given by $u_{  \mathcal{A}^{\prime}_t \setminus \mathcal{A}_{t}  } = u_{ \mathcal{A}_t \setminus \mathcal{A}^{\prime}_{t} }$.
}

\caption{Design $t$th transmission}
\label{alg1}
\end{algorithm}

Obviously, for $t=1$, the code is a normal QUP-polar code. Other codes are extended from previous codes with dynamically frozen bits. The bits which are already frozen have to be frozen for all extensions. This scheme is equivalent to the polar code extension method in \cite{ma2017incremental} if $n_1$ is a power of two and $n_t = \sum_{q=1}^{t-1}n_q,~t=2,3,\dots,t_{\text{max}}$.

For example, we consider $k=5$ and $n=\left(7,5\right)$. The information bits are $m_{1,\dots,5}$.
	\begin{itemize}
		\item For the first transmission, $\mathcal{A}_1=\left\{12,13,14,15,16\right\}$ for the $\left(7,5,16\right)$ QUP-polar code. Thus, $u_{\mathcal{A}_1}=m_1^5$ and $u_{\{1,\dots,16\}\setminus{\mathcal{A}_1}}=0$. We encode the vector by $c(1)_1^{16} = u_1^{16} \mathbb{F}^{4}$ and transmit the last 7 bits $c(1)_{10}^{16}$.
		\item For the second transmission, $\mathcal{A}_2=\left\{8,12,14,15,16\right\}$ for the $\left(12,5,16\right)$ QUP-polar code. We precode $u_1^{16}$ with $u_{\mathcal{A}_2 \setminus \mathcal{A}_1} = u_{\mathcal{A}_1 \setminus \mathcal{A}_2}$, which is $u_8=u_{13}=m_2$ in this case. We encode $c(2)_1^{16} = u_1^{16} \mathbb{F}^{4}$ and send $c(2)_5^{9}$. 
		
		Because $\mathbb{F}^{4}$ is a lower triangular matrix, $u_8=u_{13}$ does not change the first transmitted bits, i.e., $c(2)_{10}^{16} = c(1)_{10}^{16}$. At receiver, we decode the $\left(12,5,16\right)$ QUP-polar code from the noisy version of $c(2)_{5}^{16}$. Note that $u_{13}$ is now a dynamically frozen bit with constraint $u_{13}=u_8$.
	\end{itemize}
In this example, the QUP-polar codes with optimal information set are decoded after every single transmission by using all received information.

\subsection{Extension to MLC}
With $2^m$-ASK, $n_t$ symbols ($mn_t$ bits) are transmitted in the $t$th transmission. We replace line~1,2 in Algorithm~\ref{alg1} with an MLPC construction algorithm (MI-DGA in this work).

\section{Design Examples and Simulation Results}\label{sec:sim}
\begin{figure*}
	\centering
	\begin{tikzpicture}[scale=1]
	\footnotesize
	\begin{semilogyaxis}[
	width=6.5in,
	height=2.5in,
	legend style={at={(0,0)},anchor=south west},
	ymin=0.0001,
	ymax=1,
	minor x tick num=4,
	minor y tick num=5,
	major grid style={thick},
	grid=both,
	xmin = -6,
	xmax =3.5,
	xlabel = SNR in dB,
	ylabel = FER,
	mark options={solid},
	]
	\addplot[red, mark = square]
	table[x=snr,y=fer]{snr fer
		-7 0.9009
		-6.75 0.78125
		-6.5 0.68027
		-6.25 0.55866
		-6 0.45662
		-5.75 0.34965
		-5.5 0.17825
		-5.25 0.11628
		-5 0.063331
		-4.75 0.026116
		-4.5 0.010002
		-4.25 0.004745
		-4 0.00093231
	}; 
	
\addplot[blue, mark = square]
table[x=snr,y=fer]{snr fer
	-7 0.97087
	-6.75 0.98039
	-6.5 0.90909
	-6.25 0.81301
	-6 0.81967
	-5.75 0.77519
	-5.5 0.61728
	-5.25 0.48544
	-5 0.32787
	-4.75 0.25974
	-4.5 0.13004
	-4.25 0.070323
	-4 0.036563
	-3.75 0.012355
	-3.5 0.0045637
	-3.25 0.0014041
	-3 0.0004238
};

\addplot[orange, mark = square]
table[x=snr,y=fer]{snr fer
	-5 0.93458
	-4.75 0.91743
	-4.5 0.91743
	-4.25 0.79365
	-4 0.68966
	-3.75 0.55556
	-3.5 0.49261
	-3.25 0.25
	-3 0.1675
	-2.75 0.097752
	-2.5 0.046773
	-2.25 0.021935
	-2 0.0080457
	-1.75 0.0024843
	-1.5 0.00067897
};

\addplot[black, mark = square]
table[x=snr,y=fer]{snr fer
	-1 0.9434
	-0.75 0.93458
	-0.5 0.82645
	-0.25 0.62893
	0 0.4902
	0.25 0.4717
	0.5 0.3012
	0.75 0.19608
	1 0.1227
	1.25 0.055586
	1.5 0.026788
	1.75 0.0088582
	2 0.0032242
	2.25 0.0010545
	2.5 0.000194
};
	
\addplot[red, mark = x]
table[x=snr,y=fer]{snr fer
	-6 0.96154
	-5.75 0.89286
	-5.5 0.81301
	-5.25 0.68966
	-5 0.5814
	-4.75 0.38168
	-4.5 0.22523
	-4.25 0.14006
	-4 0.066225
	-3.75 0.030211
	-3.5 0.011743
	-3.25 0.0041748
	-3 0.0012172
	-2.75 0.00042002
}; 

\addplot[blue, mark = x]
table[x=snr,y=fer]{snr fer
	-6 0.9901
	-5.75 1
	-5.5 0.98039
	-5.25 0.97087
	-5 0.95238
	-4.75 0.83333
	-4.5 0.74074
	-4.25 0.50505
	-4 0.38462
	-3.75 0.26247
	-3.5 0.14306
	-3.25 0.08726
	-3 0.037707
	-2.75 0.017464
	-2.5 0.0052414
	-2.25 0.001637
	-2 0.0005465
};

\addplot[orange, mark = x]
table[x=snr,y=fer]{snr fer
	-4 0.97087
	-3.75 0.9434
	-3.5 0.9009
	-3.25 0.84034
	-3 0.70423
	-2.75 0.50505
	-2.5 0.44643
	-2.25 0.31546
	-2 0.16584
	-1.75 0.1001
	-1.5 0.044723
	-1.25 0.017319
	-1 0.0070681
	-0.75 0.0029273
	-0.5 0.00111
	-0.25 0.00037029
};

\addplot[black, mark =x]
table[x=snr,y=fer]{snr fer
	0 0.90909
	0.25 0.82645
	0.5 0.81301
	0.75 0.61728
	1 0.51546
	1.25 0.33784
	1.5 0.19841
	1.75 0.12547
	2 0.066845
	2.25 0.028885
	2.5 0.01459
	2.75 0.0056718
	3 0.0017645
	3.25 0.00057772
};
	
	\addplot[red,  dashed, mark=+]
	table[x=snr,y=fer]{snr fer
		-7 0.84685
		-6.75 0.875
		-6.5 0.68707
		-6.25 0.50838
		-6 0.43379
		-5.75 0.31119
		-5.5 0.17647
		-5.25 0.10581
		-5 0.051298
		-4.75 0.022983
		-4.5 0.0085017
		-4.25 0.0029419
		-4 0.00083908
	}; 
	
	\addplot[blue,  dashed, mark=+]
	table[x=snr,y=fer]{snr fer
		-7 0.97087
		-6.75 0.98039
		-6.5 0.90909
		-6.25 0.81301
		-6 0.81967
		-5.75 0.77519
		-5.5 0.61728
		-5.25 0.48544
		-5 0.32787
		-4.75 0.25974
		-4.5 0.13004
		-4.25 0.070323
		-4 0.036563
		-3.75 0.012355
		-3.5 0.0045637
		-3.25 0.0014041
		-3 0.0004238
	};
	
	\addplot[orange,  dashed, mark=+]
	table[x=snr,y=fer]{snr fer
		-5 0.93458
		-4.75 0.91743
		-4.5 0.91743
		-4.25 0.79365
		-4 0.68966
		-3.75 0.55556
		-3.5 0.49261
		-3.25 0.25
		-3 0.1675
		-2.75 0.097752
		-2.5 0.046773
		-2.25 0.021935
		-2 0.0080457
		-1.75 0.0024843
		-1.5 0.00067897
	};
	
	\addplot[black, dashed, mark=+]
	table[x=snr,y=fer]{snr fer
		-1 0.9434
		-0.75 0.93458
		-0.5 0.82645
		-0.25 0.62893
		0 0.4902
		0.25 0.4717
		0.5 0.3012
		0.75 0.19608
		1 0.1227
		1.25 0.055586
		1.5 0.026788
		1.75 0.0088582
		2 0.0032242
		2.25 0.0010545
		2.5 0.000194
	};
	
		\addplot[red,  mark=*]
		table[x=snr,y=fer]{snr fer
			-8 1
			-7.75 1
			-7.5 1
			-7.25 1
			-7 1
			-6.75 1
			-6.5 0.96154
			-6.25 0.80645
			-6 0.81967
			-5.75 0.7814
			-5.5 0.69444
			-5.25 0.45045
			-5 0.2809
			-4.75 0.22624
			-4.5 0.089445
			-4.25 0.038555
			-4 0.019275
			-3.75 0.0050551
			-3.5 0.0017155
			-3.25 0.00049075
		}; 
		
		\addplot[blue, mark=*]
		table[x=snr,y=fer]{snr fer
			-7 1
			-6.75 1
			-6.5 1
			-6.25 0.98039
			-6 0.98039
			-5.75 0.98039
			-5.5 1
			-5.25 0.86207
			-5 0.76923
			-4.75 0.68493
			-4.5 0.50505
			-4.25 0.34247
			-4 0.22523
			-3.75 0.12285
			-3.5 0.063291
			-3.25 0.027307
			-3 0.010356
			-2.75 0.0039355
			-2.5 0.0011376
			-2.25 0.00028464
		};
		
		\addplot[orange , mark=*]
		table[x=snr,y=fer]{snr fer
			-4 0.96154
			-3.75 0.92593
			-3.5 0.71429
			-3.25 0.69444
			-3 0.56818
			-2.75 0.46296
			-2.5 0.27933
			-2.25 0.19685
			-2 0.091241
			-1.75 0.035063
			-1.5 0.022883
			-1.25 0.0068409
			-1 0.0025272
			-0.75 0.00091795
		};
		
		\addplot[black,  mark=*]
		table[x=snr,y=fer]{snr fer
			-1 1
			-0.75 1
			-0.5 1
			-0.25 0.98039
			0 0.96154
			0.25 0.84746
			0.5 0.80645
			0.75 0.625
			1 0.46765
			1.25 0.35971
			1.5 0.2439
			1.75 0.12755
			2 0.07764
			2.25 0.044723
			2.5 0.014573
			2.75 0.007677
			3 0.0023702
			3.25 0.0007919
		};
	
	\end{semilogyaxis}
	\end{tikzpicture}
	\caption{biAWGN, $k=128$, $n_1^4=\left(250,250,200,140\right)$, design SNR $=\left(3, -1, -2.5, -3\right)$dB}
	\label{fig:ex1}
\end{figure*}
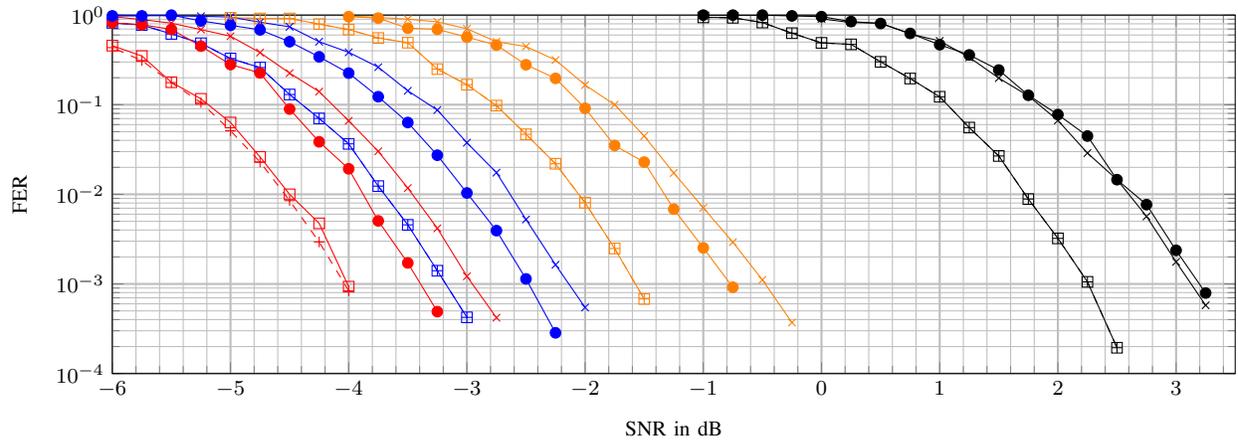

\begin{figure*}
	\centering
	\begin{tikzpicture}[scale=1]
	\footnotesize
	\begin{semilogyaxis}[
	width=6.5in,
	height=2.5in,
	legend style={at={(0,0)},anchor=south west},
	ymin=0.0001,
	ymax=1,
	minor x tick num=4,
	minor y tick num=5,
	major grid style={thick},
	grid=both,
	xmin = -5,
	xmax =7,
	xlabel = SNR in dB,
	ylabel = FER,
	mark options={solid},
	]
	\addplot[red, mark = square]
	table[x=snr,y=fer]{snr fer
		-5 0.9434
		-4.75 0.84746
		-4.5 0.57471
		-4.25 0.29412
		-4 0.12821
		-3.75 0.020747
		-3.5 0.0035656
		-3.25 0.00032664
	}; 
	
	\addplot[blue, mark = square]
	table[x=snr,y=fer]{snr fer
		-4 0.96154
		-3.75 0.86207
		-3.5 0.55556
		-3.25 0.29586
		-3 0.097466
		-2.75 0.018477
		-2.5 0.0028014
		-2.25 0.00031058
	};
	
	\addplot[orange, mark = square]
	table[x=snr,y=fer]{snr fer
		-1 0.93458
		-0.75 0.83333
		-0.5 0.54348
		-0.25 0.2457
		0 0.085106
		0.25 0.017464
		0.5 0.0021541
		0.75 0.0002194
	};
	
	\addplot[black, mark = square]
	table[x=snr,y=fer]{snr fer
		4 0.9901
		4.25 0.98039
		4.5 0.82645
		4.75 0.71429
		5 0.42553
		5.25 0.18519
		5.5 0.048948
		5.75 0.010033
		6 0.0017573
		6.25 0.000158
	};
	
	\addplot[red, mark = x]
	table[x=snr,y=fer]{snr fer
		-5 1
		-4.75 1
		-4.5 1
		-4.25 0.96154
		-4 0.66667
		-3.75 0.36232
		-3.5 0.093284
		-3.25 0.014401
		-3 0.0018417
		-2.75 0.00019284
	}; 
	
	\addplot[blue, mark = x]
	table[x=snr,y=fer]{snr fer
		-4 1
		-3.75 1
		-3.5 0.95238
		-3.25 0.88496
		-3 0.60976
		-2.75 0.30303
		-2.5 0.062696
		-2.25 0.012644
		-2 0.0018859
		-1.75 0.0002367
	};
	
	\addplot[orange, mark = x]
	table[x=snr,y=fer]{snr fer
		-1 1
		-0.75 0.95238
		-0.5 0.83333
		-0.25 0.55249
		0 0.17699
		0.25 0.035063
		0.5 0.0057376
		0.75 0.00045526
	};
	
	\addplot[black, mark =x]
	table[x=snr,y=fer]{snr fer
		4 1
		4.25 1
		4.5 0.9901
		4.75 0.95238
		5 0.80645
		5.25 0.49505
		5.5 0.21786
		5.75 0.07758
		6 0.020084
		6.25 0.005161
		6.5 0.0010534
		6.75 0.00028896
	};
	
	\addplot[red,  dashed, mark=+]
	table[x=snr,y=fer]{snr fer
		-5 0.9434
		-4.75 0.88136
		-4.5 0.62069
		-4.25 0.27647
		-4 0.079487
		-3.75 0.019087
		-3.5 0.0027098
		-3.25 0.00024825
	}; 
	
	\addplot[blue,  dashed, mark=+]
	table[x=snr,y=fer]{snr fer
		-4 0.96154
		-3.75 0.86207
		-3.5 0.55556
		-3.25 0.29586
		-3 0.097466
		-2.75 0.018477
		-2.5 0.0028014
		-2.25 0.00031058
		
	};
	
	\addplot[orange,  dashed, mark=+]
	table[x=snr,y=fer]{snr fer
		-1 0.93458
		-0.75 0.83333
		-0.5 0.54348
		-0.25 0.2457
		0 0.085106
		0.25 0.017464
		0.5 0.0021541
		0.75 0.0002194
	};
	
	\addplot[black, dashed, mark=+]
	table[x=snr,y=fer]{snr fer
		4 0.9901
		4.25 0.98039
		4.5 0.82645
		4.75 0.71429
		5 0.42553
		5.25 0.18519
		5.5 0.048948
		5.75 0.010033
		6 0.0017573
		6.25 0.000158
	};
	
	\addplot[red,  mark=*]
	table[x=snr,y=fer]{snr fer
		-5 0.98039
		-4.75 0.89286
		-4.5 0.50505
		-4.25 0.20161
		-4 0.039777
		-3.75 0.0044912
		-3.5 0.00023132
	}; 
	
	\addplot[blue, mark=*]
	table[x=snr,y=fer]{snr fer
		-4 0.98039
		-3.75 0.87719
		-3.5 0.65789
		-3.25 0.27322
		-3 0.068027
		-2.75 0.0094038
		-2.5 0.00059627
	};
	
	\addplot[orange , mark=*]
	table[x=snr,y=fer]{snr fer
		-1 0.96154
		-0.75 0.92593
		-0.5 0.80645
		-0.25 0.33784
		0 0.11364
		0.25 0.01547
		0.5 0.0018484
		0.75 0.00014
	};
	
	\addplot[black,  mark=*]
	table[x=snr,y=fer]{snr fer
		4 1
		4.25 1
		4.5 1
		4.75 0.96154
		5 0.86207
		5.25 0.51546
		5.5 0.28571
		5.75 0.091075
		6 0.030921
		6.25 0.0069367
		6.5 0.0017316
		6.75 0.00055224
	};
	
	\end{semilogyaxis}
	\end{tikzpicture}
	\caption{biAWGN, $k=848$, $n_1^4=\left(1000,1000,1500,800\right)$, design SNR $=\left(6.5, 1, -2, -3\right)$dB}
	\label{fig:ex2}
\end{figure*}
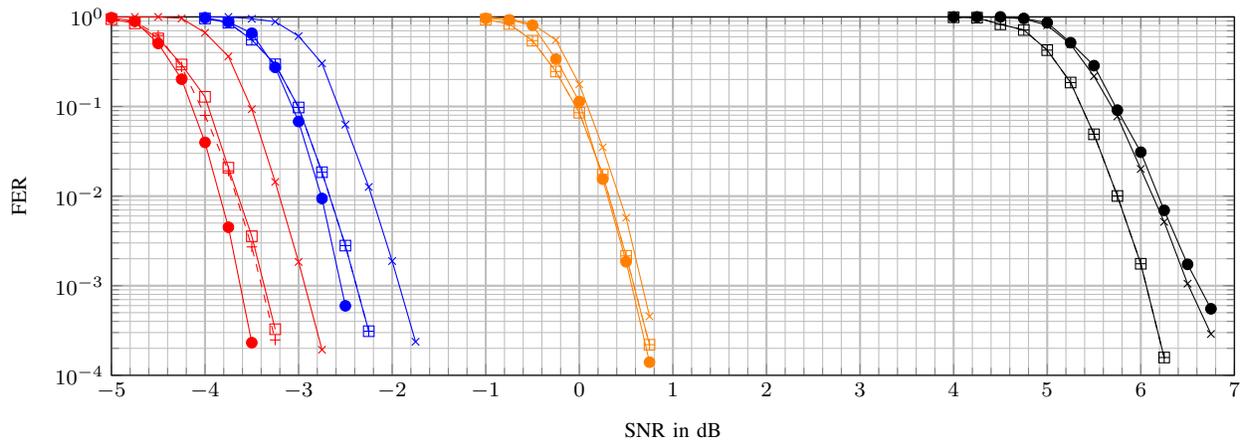

\begin{figure*}
	\centering
	\footnotesize
	\begin{tikzpicture}[scale=1]
	\begin{semilogyaxis}[
	width=6.5in,
	height=2.5in,
	legend style={at={(0,0)},anchor=south west},
	ymin=0.0001,
	ymax=1,
	minor x tick num=4,
	minor y tick num=5,
	major grid style={thick},
	grid=both,
	xmin = 2,
	xmax =17,
	xlabel = SNR in dB,
	ylabel = FER,
	mark options={solid},
	]
	\addplot[red, mark = square]
	table[x=snr,y=fer]{snr fer
		2 1
		2.25 0.98039
		2.5 0.98039
		2.75 0.9434
		3 0.64935
		3.25 0.36496
		3.5 0.12531
		3.75 0.04686
		4 0.0089993
		4.25 0.00099737
	};\label{scl4} 
	
	\addplot[blue, mark = square]
	table[x=snr,y=fer]{snr fer
		4 1
		4.25 0.96154
		4.5 0.92593
		4.75 0.64935
		5 0.45045
		5.25 0.15773
		5.5 0.048876
		5.75 0.010462
		6 0.0010491
		6.25 0.0001262
	};\label{scl3}
	
	\addplot[orange, mark = square]
	table[x=snr,y=fer]{snr fer
		8.5 1
		8.75 0.9434
		9 0.86207
		9.25 0.68493
		9.5 0.43103
		9.75 0.18657
		10 0.071124
		10.25 0.018162
		10.5 0.0022694
		10.75 0.00027466
	};\label{scl2}
	
	\addplot[black, mark = square]
	table[x=snr,y=fer]{snr fer
		13.5 1
		13.75 0.9434
		14 0.92593
		14.25 0.64935
		14.5 0.42017
		14.75 0.22422
		15 0.07278
		15.25 0.022563
		15.5 0.0040199
		15.75 0.0005
	};\label{scl1}
	
	\addplot[red, mark = x]
	table[x=snr,y=fer]{snr fer
		4.75 1
		5 0.89286
		5.25 0.73529
		5.5 0.43478
		5.75 0.17606
		6 0.056948
		6.25 0.011696
		6.5 0.0024394
		6.75 0.00035196
	};\label{turbo4}
	
	\addplot[blue, mark = x]
	table[x=snr,y=fer]{snr fer
		5 1
		5.25 1
		5.5 0.96154
		5.75 0.98039
		6 0.73529
		6.25 0.55556
		6.5 0.25
		6.75 0.072993
		7 0.017403
		7.25 0.0027276
		7.5 0.00030549
	};\label{turbo3}
	
	\addplot[orange, mark = x]
	table[x=snr,y=fer]{snr fer
		9.25 1
		9.5 0.98039
		9.75 0.89286
		10 0.71429
		10.25 0.5102
		10.5 0.2381
		10.75 0.067385
		11 0.015728
		11.25 0.0037075
		11.5 0.00052576
	};\label{turbo2}
	
	\addplot[black, mark =x]
	table[x=snr,y=fer]{snr fer
		14.25 1
		14.5 0.98039
		14.75 0.81967
		15 0.65789
		15.25 0.36496
		15.5 0.18939
		15.75 0.085616
		16 0.029002
		16.25 0.006
		16.5 0.001
	};\label{turbo1}
	
	\addplot[red,  dashed, mark=+]
	table[x=snr,y=fer]{snr fer
		2.5 1
		2.75 0.83019
		3 0.49351
		3.25 0.28467
		3.5 0.11028
		3.75 0.034677
		4 0.0039597
		4.25 0.00083779
		
	};\label{sc4} 
	
	\addplot[blue,  dashed, mark=+]
	table[x=snr,y=fer]{snr fer
		4 1
		4.25 0.96154
		4.5 0.88889
		4.75 0.62338
		5 0.34234
		5.25 0.16088
		5.5 0.033236
		5.75 0.0064867
		6 0.00077635
		6.25 7.0673e-05
	};\label{sc3}
	
	\addplot[orange,  dashed, mark=+]
	table[x=snr,y=fer]{snr fer
		8.5 1
		8.75 0.9434
		9 0.86207
		9.25 0.68493
		9.5 0.43103
		9.75 0.18657
		10 0.071124
		10.25 0.018162
		10.5 0.0022694
		10.75 0.00027466
	};\label{sc2}
	
	\addplot[black, dashed, mark=+]
	table[x=snr,y=fer]{snr fer
		13.5 1
		13.75 0.9434
		14 0.92593
		14.25 0.64935
		14.5 0.42017
		14.75 0.22422
		15 0.07278
		15.25 0.022563
		15.5 0.0040199
		15.75 0.0005
	};\label{sc1}
	
	\addplot[red,  mark=*]
	table[x=snr,y=fer]{snr fer
		2.75 1
		3 0.98039
		3.25 0.9434
		3.5 0.9434
		3.75 0.66667
		4 0.3937
		4.25 0.20243
		4.5 0.049407
		4.75 0.0070902
		5 0.0014867
		5.25 0.00015204
	};\label{ldpc4} 
	
	\addplot[blue, mark=*]
	table[x=snr,y=fer]{snr fer
		5 1
		5.25 0.9434
		5.5 0.9434
		5.75 0.70423
		6 0.4902
		6.25 0.16611
		6.5 0.060386
		6.75 0.010811
		7 0.003175
		7.25 0.00041528
	};\label{ldpc3}
	
	\addplot[orange , mark=*]
	table[x=snr,y=fer]{snr fer
		9.5 1
		9.75 0.92593
		10 0.83333
		10.25 0.55556
		10.5 0.24876
		10.75 0.097847
		11 0.030998
		11.25 0.0069004
		11.5 0.0011555
		11.75 0.00017056
	};\label{ldpc2}
	
	\addplot[black,  mark=*]
	table[x=snr,y=fer]{snr fer
		14 0.98039
		14.25 0.98039
		14.5 0.81967
		14.75 0.64103
		15 0.36496
		15.25 0.16722
		15.5 0.073099
		15.75 0.015913
		16 0.0038084
		16.25 0.0010253
		16.5 0.00022601
	};\label{ldpc1}
	
	\end{semilogyaxis}
	\end{tikzpicture}
	\caption{8-ASK, $k=896$, $n_1^4=\left(1200,600,1200,900\right)$, design SNR $=\left(16.25, 11.25, 6.75, 5\right)$dB}
	\label{fig:ex3}
\end{figure*}
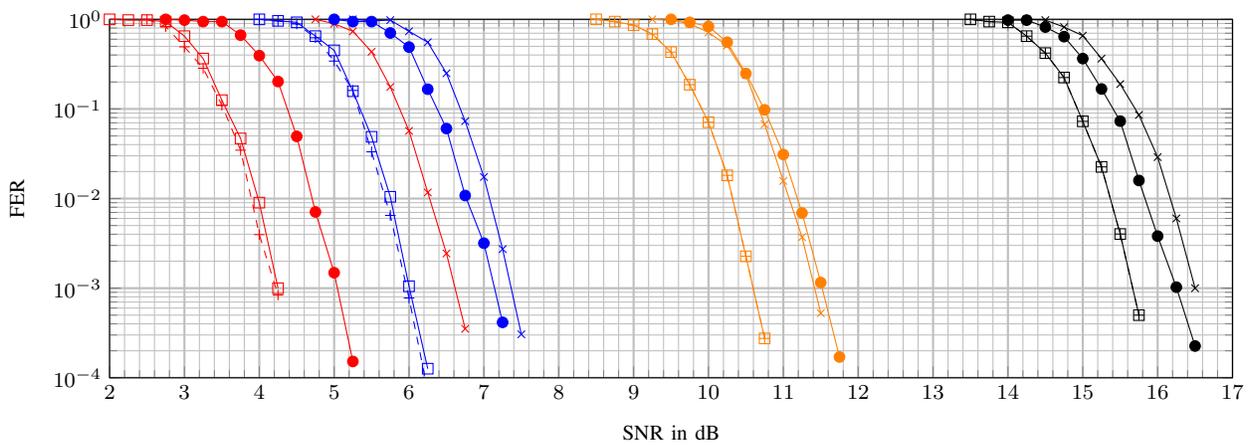

\begin{table}
	\footnotesize
	\caption{Explanation for the curves in Fig.~\ref{fig:ex1}, Fig.~\ref{fig:ex2}, Fig.~\ref{fig:ex3}. }
	\label{tab:label}
	\centering
	{\renewcommand{\arraystretch}{1.5}
		\begin{tabular}{|c||c|c|c|c|}				
			\hline
			& direct polar& proposed polar & 5G-LDPC& LTE-turbo\\
			\hline
			$1$st & \ref{sc1}&  \ref{scl1} &  \ref{ldpc1} &  \ref{turbo1}\\
			\hline
			$2$nd & \ref{sc2}&  \ref{scl2} &  \ref{ldpc2} &  \ref{turbo2}\\
			\hline
			$3$rd & \ref{sc3}&  \ref{scl3} &  \ref{ldpc3} &  \ref{turbo3}\\
			\hline
			$4$th & \ref{sc4}&  \ref{scl4} &  \ref{ldpc4} &  \ref{turbo4}\\
			\hline
		\end{tabular}
	}
\end{table}	
In this section, three design examples for $t_{\text{max}}=4$ are shown in Fig.~\ref{fig:ex1}, Fig.~\ref{fig:ex2}, Fig.~\ref{fig:ex3}. We use 16 bits CRC with generator polynomial \enquote{$x^{16}+x^{12}+x^5+1$} for error detection. The polar codes are decoded by min-sum approximated SCL decoding with list size 32. Log-MAP decoding with 10 iterations and belief propagation (BP) with 50 iterations are used for LTE-turbo and 5G-LDPC codes, respectively. In the 8-ASK example, Bit-Interleaved Coded Modulation (BICM) \cite{BICM} is used for turbo and LDPC codes. Note that the directly designed polar codes (dashed curves) are $\left(\sum_{q=1}^{t}n_q,k,\cdot\right)$ QUP-polar codes only serve as a reference and can not work for an IR-HARQ scheme. 

The simulation results show that the polar codes generated by the proposed algorithm perform very similar to directly designed polar codes. In the 8-ASK example, the proposed scheme performs approximately $1$dB better than 5G-LDPC codes after two to four transmissions. 

Because of the extension constraint, it is hard to extend a heavily punctured polar codes. Consider an $\left(N+w,k,2N\right)$ QUP-polar code, where $w$ is a positive integer and $w\ll N$. Normally, the first $w$ bits are unreliable, while these bits are almost perfect in the extended $(2N,k)$ code. This effect degrades the performance for all further extensions. For example, for $n_1^4=\left(1000,100,100,100\right)$, the $3$rd and $4$th polar codes perform much worse than directly designed code because $\sum_{q=1}^{2}n_q=1100$ and $1100-1024=76\ll1024$. Therefore, we should avoid using heavily punctured polar code for the $\left\{1,\dots,t_\text{max}-1\right\}$th transmissions in biAWGN channel. However, this effect disappears for MLPC. We can design very good codes for $n_1^4/m=\left(1000,100,100,100\right)$. The reason should be the automatically controlled code rate for $m$ polar codes.

\section{Conclusion}\label{sec:con}
In this paper, an IR-HARQ scheme based on QUP and dynamically frozen bits for biAWGN channel and PCM is proposed. Simulation results show that the rate-matched polar codes generated by the proposed algorithm perform very similar to directly designed QUP-polar codes.
 
For future work, this scheme can be applied for a fading channel, i.e. the channel information estimated by training symbols of the previous transmissions could be used to design polar codes for next transmission.

\bibliographystyle{IEEEtran}
\bibliography{IEEEabrv,references}
\end{document}